\documentclass[11 pt]{amsart}
\usepackage{amsmath}
\usepackage{amssymb}
\usepackage{amsthm}
\usepackage{dsfont}
\usepackage{fancyhdr}
\usepackage[all]{xy}
\usepackage{hyperref}
\usepackage{paralist}
\usepackage{tikz}
\usepackage{verbatim}
\usepackage[hmargin=1.2in,vmargin=1.2in]{geometry}

\newtheorem{theorem}{Theorem}[section]
\newtheorem{lemma}[theorem]{Lemma}
\theoremstyle{proposition}

\theoremstyle{corollary}

\theoremstyle{definition}

\numberwithin{equation}{section}

\newcommand{\cE}{{\mathcal E}}

\newcommand{\cF}{{\mathcal F}}
\newcommand{\cG}{{\mathcal G}}
\newcommand{\cH}{{\mathcal H}}

\newcommand{\Z}{\mathbb Z}

\newcommand{\F}{\mathbb F}

\newcommand{\bt}{{\bf \theta}}

\newcommand{\la}{{\langle}}
\newcommand{\ra}{{\rangle}}

\begin{document}
\title[]{Last fall degree of semi-local polynomial systems}
\author[]{Ming-Deh A. Huang (USC, mdhuang@usc.edu)}
\address{Computer Science Department,University of Southern California, U.S.A.}
\email{mdhuang@usc.edu}

\urladdr{}

\begin{abstract}
 We study the last fall degrees of {\em semi-local} polynomial systems, and the computational complexity of solving such systems for closed-point and rational-point solutions, where the systems are defined over a finite field.
A semi-local polynomial system specifies an algebraic set which is the image of a global linear transformation of a direct product of local affine algebraic sets.  As a special but interesting case, polynomial systems that arise from Weil restriction of algebraic sets in an affine space of low dimension are semi-local.  Such systems have received considerable attention due to their application in cryptography. Our main results bound the last fall degree of a semi-local polynomial system in terms of the number of closed point solutions, and yield an efficient algorithm for finding all rational-point solutions when the prime characteristic of the finite field and the number of rational solutions are small.
Our results on solving semi-local systems imply an improvement on a previously known polynomial-time attack on the HFE (Hidden Field Equations) cryptosystems. The attacks implied in our results extend to public key encryption functions which are based on semi-local systems where either the number of closed point solutions is small, or the characteristic of the field is small.  It remains plausible  to construct public key cryptosystems based on semi-local systems over a finite field of large prime characteristic with exponential number of closed point solutions.  Such a method is presented in the paper, followed by further cryptanalysis involving the isomorphism of polynomials (IP) problem, as well as a concrete public key encryption scheme which is secure against all the attacks discussed in this paper.

\end{abstract}

\maketitle

\section{Introduction}
In this paper, we study the last fall degrees of {\em semi-local} polynomial systems, and the computational complexity of solving such systems, for closed point as well as $k$-rational points when the systems are defined over a finite field $k$.
A semi-local polynomial system specifies an algebraic set which is the image of a (global) linear transformation of a direct product of (local) affine algebraic sets.  As a special but interesting case, polynomial systems that arise from Weil restriction of algebraic sets in an affine space of low dimension are semi-local.  The computational complexity of solving polynomial systems arising from Weil restriction have received considerable attention due to their application in cryptography
(see for examples \cite{BET,DIN,FAU3,GRA,HKY,HKYY}).

Suppose $\cF=\{ F_1,\ldots, F_m\}$ where $F_i \in k[x_1,\ldots,x_n]$ and $k$ is a field.  We often refer to $\cF$ as a polynomial system (consisting of $m$ polynomials in $n$ variables).  We say the system is square if $n=m$.  For $\mu\in Gl_m (k)$ represented by $m$ by $m$ matrix
$(a_{ij})_{1\le i,j \le m}$ with $a_{ij}\in k$, let $\mu\circ\cF$ denote the ordered set $\{ G_1,\ldots G_m\}$ where $G_i = \sum_{j=1}^m a_{ij} F_j$.
For $\lambda\in Gl_n (k)$ represented by $n$ by $n$ matrix  $(\alpha_{ij})_{1\le i,j \le n}$ with $\alpha_{ij}\in k$, let $\cF\circ\lambda$ denote the ordered set of polynomials
$\{ F_1\circ\lambda,\ldots,F_m\circ\lambda\}$, where for $F\in k[x_1,\ldots,x_n]$, $F\circ\lambda = F (L_1,\ldots, L_n)$ with $L_i = \sum_{j=1}^n \alpha_{ij} x_j$.

We say that $\cF$ is $c$-{\em local} if $n=c\ell$ for some positive integer $\ell$, and $\cF$ can be partitioned into subsets $\cF_1,\ldots,\cF_{\ell}$ where the polynomials in $\cF_i$ are  polynomials in $x_{ci-i+1},\ldots, x_{ci}$ for $i=1,\ldots,\ell$.  We say that an ordered set of polynomials $\cG=\{G_1,\ldots,G_m\}$ is $c$-{\em semi-local} if $\cG = \mu\circ\cF\circ\lambda$ for some $\mu\in Gl_m (k)$,$ \lambda\in Gl_n (k)$, $\cF\subset k[x_1,\ldots,x_n]$, and $\cF$ is $c$-local.

Suppose $K$ is a finite extension of a finite field $k$ of degree $n$, and $V$ is an algebraic set in $\bar{k}^c$ defined over $K$.  Then a $K/k$-Weil-restriction of $V$ is linearly isomorphic over $K$ to the direct product of the conjugate algebraic sets of $V$ over $k$, hence can be described as $Z(\cG)$ for some  $c$-semi-local polynomial system $\cG$ (which we call a {\em Weil descent system}, see \S~\ref{weil} for more details).  The case where $V$ is zero-dimensional and defined by a single univariate polynomial over $K$ was studied in the context of HFE (Hidden Field Equations) cryptosystems, where decryption amounts to solving for $k$-rational solutions of a semi-local polynomial system in $n=[K:k]$ variables that describes the $K/k$-Weil restriction of $V$ \cite{BET,DIN,FAU3,GRA,HKY,HKYY}.

Suppose $\cF$ is a finite set of polynomials in $k[x_1,\ldots,x_n]$ and $k$ is a field.
The {\em first fall degree} of $\cF$ is the smallest $d\ge \deg \cF$ such that there exist $g_i\in k[x_1,\ldots,x_n]$, $f_i\in \cF$ with $\deg g_i f_i \le d$, $\sum_i g_i f_i \neq 0$, and $\deg(\sum_i g_i f_i) < d$.
The concept of first fall degree was used to heuristically bound the complexity of Gr\"{o}bner basis algorithm \cite{PQ}.
The concept of {\em last fall degree} was subsequently introduced in \cite{HKY} for deriving rigorous complexity bound in solving polynomial systems.
For $i\in\Z_{\ge 0}$, let $R_i$ denotes the set of polynomials in $k[x_1,\ldots,x_n]$ with degree no greater than $i$, and $V_i$ denote the largest subset of
the ideal generated by $\cF$ that can be constructed by doing ideal operations without exceeding degree $i$.   The last fall degree
of $\cF$ is the largest $c\in\Z_{\ge 0}$ such that $V_c \cap R_{c-1} \neq V_{c-1}$ (see \S~\ref{prep} for more details).
As shown in \cite{HKY,HKYY} the last fall degree  is intrinsic to a polynomial
system, independent of the choice of a monomial order, always bounded by the degree of regularity,  and  invariant under linear change of variables and linear change of equations.

Let $\bar{k}$ denote the algebraic closure of $k$.  We assume $\cF_i$ is zero-dimensional for all $i$, that is, the zero set $Z(\cF_i) = \{ (a_1,\ldots,a_c)\in\bar{k}^c: f(a_1,\ldots,a_c)=0$ for all $f\in \cF_i\}$ is finite.   In addition we assume that the ideal generated by $\cF_i$ is radical for all $i$.

In general a zero-dimensional polynomial system $\cG$ in $n$ variables of degree bounded by $d$ has at most $d^n$ points and can be solved in $d^{O(n)}$ time \cite{LL}.    Suppose the last fall degree of $\cG$ is $\delta$.  If $| Z(\cG)| = 1$ and the ideal generated by $\cG$ is radical, then the unique point in $Z(\cG)$ can be found in time $n^{O(\delta)}$, which is polynomial if $\delta=O(1)$ (see \S~\ref{first}).
More generally, it is an interesting question whether the running time of solving $\cG$ can depend on the cardinality of the zero set $Z(\cG)$ and the last fall degree of $\cG$.
Suppose $\cG$ is defined over a finite field $k$ and the number of $k$-rational solutions in $Z(\cG)$ is small, say bounded by a constant.  It is also an interesting question whether such solutions can be found efficiently.
It was proven in \cite{H} that the last fall degree is $O(d)$ for a 1-semi-local system $\cG=\mu\circ\cF\circ\lambda$ where $\cF$ consisting of $n$ univariate polynomials of degree bounded by $d$.  Consequently, if $| Z(\cG) | = O(1)$ and $d=O(1)$,  then  $Z(\cG)$ can be determined in $n^{O(d)} = n^{O(1)}$ time.
The last fall degree of the system $\cG'$ consisting of $\cG$ together with field equations, so that $Z(\cG') = Z_k (\cG)$, is shown to be $O(dq)$.  The field equations $x_i^q - x_i$ constrain the solutions looked for to be rational over $k$, where $| k | = q$.  It was shown that by replacing $x_i^q-x_i$ with $x_i^p - x_{i1}$, $x_{i1}^p - x_{i2}$, ..., $x_{i\ m-1}^p - x_i$, where $q=p^m$, the resulting system $\cG''$, with $Z(\cG'')$ isomorphic to $Z_k (\cG)$, has last fall degree $O(dp)$.   Consequently, if $| Z_k (\cG) | = O(1)$, then $Z_k (\cG)$ can be determined in $n^{O(dp)} = n^{O(1)}$ time if $d = O(1)$ and $p = O(1)$.

Note that the last fall degree of a 1-semi-local system $\cG$ and that of $\cG''$, where $Z(\cG'')\simeq Z_k (\cG)$, both do not depend on $n$.
In this paper we prove that this is true in general for zero-dimensional semi-local systems $\cG$ assuming the ideals generated by the local systems $\cF_i$ are radical.

In Theorem~\ref{solveG} we prove that suppose $|Z(\cG)|=s > 0$, and for $i=1,\ldots,\ell$, the ideal generated by $\cF_i$ is radical and $d_{\cF_i} \le c'$ for some constant $c'$.  Then
$d_{\cG} \le s+ c'\lceil\log_2 s\rceil$.  Moreover,
the $s$ points in $Z(\cG)$ can be found in time
$n^{O(c')}+ s^{O(c\log s)}$.   Note that this is polynomial in $n$ if $c$, $c'$, $s$ are all $O(1)$.

In Theorem~\ref{lfdk} we prove that if the ideal generated by $\cF_i$ is radical for all $i$, and $|Z_k (\cF) | = 1$, then
$d_{\cG\cup\cE'} = O(\Delta p)$, where $\Delta$ is the maximum of $\{ | Z(\cF_i)|, d_{\cF_i}: i=1,\ldots,\ell\}$.
More generally, if  $|Z_k (\cF) | =s_0$ and $\left(\begin{array}{c} s\\2 \end{array}\right) < |k|$ where $| Z(\cF_i) | \le s$ for all $i$,  then $d_{\cG\cup\cE'} = O(\max(\Delta p, s_0 p, s_0 \log s_0))$; where
$\cE'=\{x_{i0}^p-x_{i1},\ldots, x_{i \ m-1}^p-x_{i0}: i=1,\ldots,n\}$, so that $Z(\cG\cup\cE')\simeq Z_k (\cG)$.
Moreover, the $k$-rational points in $Z(\cG)$ can be found in  $(nd)^{O(\Delta p)} + s_0^{O(\log s_0)}$ time, where $|k| = p^d$.

Suppose $c$ and $\deg F$, $d_{\cF_i}$ are all $O(1)$.  Theorem~\ref{solveG} and Theorem~\ref{lfdk} together imply that $Z_k (\cG)$ can be determined efficiently if either $| Z (\cG) | = O(1)$, or both $| Z_k (\cG) |$ and $p$ are $O(1)$.

In Theorem~\ref{wd} we show that the Weil descent, with respect to a finite field extension $K$ over $k$, of $c$ polynomials in $c$ variables over $K$, is $c$-semi-local.  This paves the way for applying the results in Theorem~\ref{solveG} and \ref{lfdk} to HFE (Hidden Field Equations) cryptosystems in \S~\ref{crypto}.  We remark that the characterization of the Weil descent systems as semi-local systems may have further applications.  For example, with the characterization the results in Theorem~\ref{solveG} and \ref{lfdk} may be applied to improve the index calculus method  \cite{Ga} for solving the discrete logarithm problem of abelian varieties, where Weil descent polynomial systems are formed and solved for the purpose of finding relations.

Cryptographic applications are discussed in \S~\ref{crypto}.  In the cryptographic context, $\cG$ is the public encryption function and $\cF$ may be assumed to be known.  The secret decryption key consists of  $\mu$ and $\lambda$ such that $\cG = \mu\circ\cF\circ\lambda$.  With $\lambda$ and $\mu$ known, the decryption problem reduces to finding $k$-rational $x$ such that $z=\lambda(x)$ and $\mu^{-1} (y) = \cF (z)$.

When $\cF$ is the $K/k$-Weil descent of a univariate polynomial $f\in k[x]$ with respect to a basis $\theta_i$, $i=1,\ldots,n$, of $K$ over $k$, we refer to the system defined by $\cG$ as an HFE-system, following \cite{HKY,HKYY}.
The classic HFE-cryptosystem first introduced by Patarin \cite{P} is a special case where $f$ is the {\em extended Dembowski-Ostrom polynomial}
\[ f= \sum_{0\le i\le j < r} a_{ij} x^{q^i+q^j} +\sum_{0\le i < r} b_i x^{q^i} + c \]
with $a_{ij}, b_i, c \in K$.   The {\em HFEv-signature scheme} is constructed by introducing additional {\em Vinegar} variables,  replacing each $b_i$ with a $K$-linear form in the Vinegar variables and replacing $c$ with a quadratic form in the Vinegar variables.   A direct forgery attack amounts to solving for $x\in k^{n+e}$ such that $\cG (x)=y$, given (message) $y\in k^{n-a}$ (where $e$ is the number of Vinegar variables). Here
$\cG=\mu\circ\cF\circ\lambda$ with $\mu,\lambda,\cF$ randomly chosen in secret, where
$\lambda\in Gl_{n+e} (k)$,  $\mu\in Gl_{n-a} (k)$, and $\cF$ is the $K/k$-Weil restriction of some modified extended Dembowski-Ostrom polynomial $f$ (with the Vinegar variables treated as unknown constants in $k$ when performing the Weil restriction).
We note that our results (as well the results in \cite{GMP,HKY,HKYY}) cannot be applied for solving such polynomial systems since the systems are not zero-dimensional.
We refer to \cite{DPT} for more recent advances in breaking the HFEv-signature scheme based on MinRank type attacks.

For the general HFE-system where the encryption function $\cG$ is constructed from the Weil restriction of a univariate polynomial, given $y\in k^n$, one is interested in finding $k$-rational solutions to $\cG(x)=y$.  The solution set is $Z_k (\cG-y) = Z(\{\cG-y\}\cup\cE)$ where $\cE=\{x_i^q-x_i: i=1,\ldots,n\}$ and $k=\F_q$.  The results in \cite{HKY} bounds the last fall degree of $\{\cG-y\}\cup\cE$ by $O(dq)$, leading to an attack that runs in $n^{O(dq)}$ time.  The more recent results in \cite{GMP,HKYY} bounds the last fall degree of  $\{\cG-y\}\cup\cE$ by $O(q\log_q d)$
(with the bound in \cite{GMP} improving on the bound in \cite{HKYY} by a factor of 1/2), leading to an attack that solves the system in $n^{O(q\log_q d)}$ time.
When $d=O(1)$ and $q=O(1)$, the running time is polynomial in $n$ in these attacks.

For $y\in k^n$ such that $\cG (x)=y$ has a unique $k$-rational solution $x$, our results on solving semi-local systems imply an $n^{O(dp)}$ time bound for finding $x$.  In comparison to \cite{GMP, HKY}, our bound is better when $dp$ is substantially smaller than $q$.

An interesting question is whether more general semi-local polynomial systems, especially those which do not arise from Weil restriction, can still be used as the basis for constructing cryptosystems.  The attacks implied in our results extend to public key encryption functions which are based on semi-local systems where either the number of closed point solutions is small, or the characteristic of the field is small.  It remains plausible  to construct public key cryptosystems based on semi-local systems over a finite field of large prime characteristic with exponential number of closed point solutions.  Such a method is presented in \S~\ref{crypto}.

Solving the polynomial system $\cG(x)=y$ also reduces to finding a semi-local decomposition $\cG=\mu\circ\cF\circ\lambda$ for $\cG$.   This problem of finding the secret key $\mu$ and $\lambda$ in the cryptographic context is the semi-local case of the {\em isomorphism of polynomials} (IP) problem: Given two polynomial maps $\cF$ and $\cG$ from $k^n$ to $k^m$, given by two sets of $n$-variate polynomials $\cG = (G_1,\ldots, G_m)$ and $\cF=(F_1,\ldots,F_m)$,  to find
$\lambda\in Gl_n(k)$, $\mu \in Gl_m (k)$ such that $\cG = \mu\circ\cF\circ\lambda$.

We extend the determinant-of-Jacobian method \cite{BFP,Kay} to solve the square 1-semi-local case of the isomorphism of polynomials (IP) problem.   We also discuss how the method may be applied to mount a partial attack for square $c$-semi-local systems ($n=m$), when $c > 1$.   We show how to modify the square system construction into non-square systems to avoid the determinant-of-Jacobian attack.
We describe a concrete and simple public key encryption scheme based on a non-square 2-semi-local system. The system is secure against all the attacks discussed in this paper.  Whether an efficient attack can be found on the system is an interesting open problem that requires further investigation.

\section{Preparation}
\label{prep}
We recall the following notation and definition concerning last fall degrees \cite{HKY,HKYY}.

Suppose $\mathcal{F}$ is a finite set of polynomials in  $k[x_1,\ldots,x_n]$.
For $i \in \Z_{\geq 0}$, let $V_{\mathcal{F},i}$ be the smallest $k$-vector space such that
\begin{enumerate}
\item
$\{f \in \mathcal{F}: \deg(f) \leq i \} \subseteq V_{\mathcal{F},i}$;
\item
if $g \in V_{\mathcal{F},i}$ and if $h \in k[ x_1,\ldots,x_n]$ with $\deg(hg) \leq i$, then $hg \in V_{\mathcal{F},i}$.
\end{enumerate}

We write $f\equiv_i g \pmod{\mathcal{F}}$, for $f,g\in k[x_1,\ldots,x_n]$, if $f-g\in V_{\mathcal{F},i}$.

Put $R=k[x_1,\ldots,n]$ and let $R_i$ denote the $k$-vector space consisting of polynomials in $R$ of degree no greater than $i$ for $i \ge 0$.  The \emph{last fall degree} as defined in \cite{HKY} (see also \cite{HKYY}) is the largest $d$ such that $V_{\mathcal{F},d} \cap R_{d-1} \neq V_{\mathcal{F},d-1}$. We denote the last fall degree of $\mathcal{F}$ by $d_{\mathcal{F}}$.

The following simple result illustrates how the last fall degree can be related to the computational complexity of solving  a polynomial system.

\begin{lemma}
\label{first}
Suppose $\mathcal{F}$ is a finite set of polynomials in  $k[x_1,\ldots,x_n]$.  If the ideal generated by $\mathcal{F}$ is radical and $Z(\mathcal{F})$ has a unique point, then the unique point can be found in $n^{O(d_{\mathcal{F}})}$ time.

\end{lemma}

\begin{proof}
For $i > 0$, a linear basis for $V_{\mathcal{F},i}$ can be constructed in time $n^{O(i)}$ (see \cite{HKY,HKYY}).  If the ideal generated by $\mathcal{F}$ is radical and $Z(\mathcal{F})$ has a unique point $(\alpha_1,\ldots, \alpha_n)\in\bar{k}^n$, then $x_i - \alpha_i \in V_{\mathcal{F}, d_{\mathcal{F}} }$ for $i=1,\ldots, n$.  It follows that $(\alpha_1,\ldots,\alpha_n)$ can be found in $n^{O(d_{\mathcal{F}})}$ time.
\end{proof}

To prove the main results on last fall degrees and the computational complexity of  solving semi-local polynomial systems, we need some technical preparation.

First, let us recall the well-known {\em Shape Lemma} (Theorem 3.7.25 of \cite{KR}):

\begin{lemma}
Let $k$ be a field and let $\cF=\{f_1,\ldots,f_m\}\subset k[x_1,\ldots,x_n]$.  Let $I$ be the ideal generated by $\cF$.  Suppose $Z(\cF)$ is finite, and suppose for any two distinct $(a_1,\ldots,a_n), (b_1,\ldots,b_n)\in Z(\cF)$ we have $a_n\neq b_n$.  Then the radical ideal of $I$ has a reduced lexicographic Gr\"{o}bner basis of the form
$$\{ g_n(x_n), x_i - g_i (x_n): i=1,\ldots, n-1\}$$
with $g_i$ univariate for all $i$ and $\deg g_i < \deg g_n = |Z(\cF)|$ for $i=1,\ldots, n-1$.
\end{lemma}

\begin{lemma}
\label{x-g}
Let $k$ be a finite field.  Suppose $\cF$ is a finite set in $k[x_1,\ldots,x_n]$ and the zero set $Z(\cF)$ is finite with $| Z(\cF) | = s$.  If $K$ is a finite field containing $k$ and $\left(\begin{array}{c} s\\2 \end{array}\right) < |K |$, then there exists $t = \sum_{i=1}^n a_i x_i \in K [x_1,\ldots,x_n]$ with $a_i \in K$ for $i=1,\ldots,n$, such that  for all $\alpha,\beta\in Z(\cF)$, if $\alpha\neq \beta$ then $t(\alpha)\neq t(\beta)$.  Moreover suppose without loss of generality $a_n\neq 0$, then there exist $g_1,\ldots, g_n \in K[t]$ with $\deg g_i < \deg g_n = s$ for $i=1,\ldots,n-1$, such that $Z(\cF)=Z(\cG)$ where $\cG=\{ g_n(t), x_i - g_i (t): i=1,\ldots,n-1\}$. In particular if the ideal generated by $\cF$ is radical then $\cG$ generates the same ideal.

\end{lemma}

\begin{proof}  For $\gamma=(\gamma_i)_{i=1}^n \in K^n$ let $t_{\gamma}=\sum_{i=1}^n \gamma_i x_i$.    For each pair $\alpha,\beta\in Z(\cF)$ with $\beta\neq \alpha$, $\{\gamma\in K^n: t_{\gamma} (\beta) = t_{\gamma} (\alpha) \}$ is a $K$-vector space of dimension $n-1$, with cardinality $q^{n-1}$ where $q=|K|$.  Therefore if $\left(\begin{array}{c} s\\2 \end{array}\right) < |K |$, then there exists $t = \sum_{i=1}^n a_i x_i \in K [x_1,\ldots,x_n]$ with $a_i \in K$ for $i=1,\ldots,n$, such that for all $\alpha,\beta\in Z(\cF)$, if $\alpha\neq \beta$ then $t(\alpha)\neq t(\beta)$.

The rest of the lemma follows by applying the Shape Lemma to $k[x_1,\ldots,x_{n-1}, t]$, noting that $k[x_1,\ldots, x_{n-1}, t]=k[x_1,\ldots, x_n]$.

\end{proof}

\begin{lemma}
\label{red}
Suppose $\mathcal{F}$ is a finite set of polynomials in  $R=k[x_1,\ldots,x_n]$.  Suppose $f\in k[x_1]$, $\deg f \le d$, and $f \equiv_d 0 \pmod{\cF}$.  Suppose moreover
$x_i \equiv_d h_i  \pmod{\cF}$ where $ h_i \in k[x_1]$,  and $\deg h_i < \deg f$ for $i=2,\ldots,n$.  Then for all monomials $m\in R$, $m \equiv_D h\pmod{\cF}$ for some $h\in k[x_1]$ with $\deg h < \deg f$ and $D=\max(\deg m, 2d)$.
\end{lemma}

\begin{proof}  We prove by induction on the degree of $m$.  The case $\deg m = 1$ is trivial.  Inductively for $m$ of degree greater than 1, write $m = m_1 m_2$ where $m_1$ is a monomial of degree $\lfloor \frac{\deg m}{2}\rfloor$. Let $d_1 = \deg m_1$ and $d_2 = \deg m_2$.  Then by induction $m_i\equiv_{\Delta_i} g_i (x_1)\pmod{\cF}$ where $\deg g_i < \deg f$ and $\Delta_i = \max (d_i, 2d)$ for $i=1,2$.  Let $\alpha_i = m_i - g_i$ for $i=1,2$.  Then $\deg \alpha_i \le \max (d_i, d)$ for $i=1,2$.
We have
\[ m - g_1 g_2 = \alpha_1 \alpha_2 + \alpha_1 g_2 +\alpha_2 g_1.\]
We have \[g_1 g_2 \equiv_{2d} g \pmod{f}\] where $g\in k[x_1]$ and $\deg g < \deg f \le d$.
Since $f\equiv_{d} 0\pmod{\cF}$, it follows that $g_1 g_2\equiv_{2d} g \pmod{\cF}$.

\ \\Case 1: $\deg m \le 2d$ then $d_i \le d$ for $i=1,2$.  In this case $ \alpha_1 \alpha_2$, $\alpha_1 g_2$, and $\alpha_2 g_1$ are all of degree no greater than $2d$.
Since $\Delta_i = 2d$, $\alpha_1,\alpha_2 \in V_{\cF, 2d}$, so $\alpha_1 \alpha_2 + \alpha_1 g_2 +\alpha_2 g_1\in V_{\cF, 2d}$.
We get \[ m=m_1 m_2 \equiv_{2d} g_1 g_2 \equiv_{2d} g \pmod{\cF}.\]

\ \\Case 2: $\deg m > 2d$.  In this case $d_i \ge d$ and $\deg \alpha_i \le \max(d_i, d) = d_i$ for $i=1,2$. So $ \alpha_1 \alpha_2$, $\alpha_1 g_2$, and $\alpha_2 g_1$ are all of degree no greater than $d_1 + d_2 = \deg m$.  Since $\Delta_i < \deg m$, $\alpha_i\in V_{\cF,\deg m}$ for $i=1,2$.  It follows that
 $ \alpha_1 \alpha_2$, $\alpha_1 g_2$, and $\alpha_2 g_1$ are all in $V_{\cF,\deg m}$.  Hence
 \[ m\equiv_{\deg m} g_1g_2 \equiv_{2d} g \pmod{\cF}\]
where $g\in k[x_1]$ and $\deg g < \deg f \le d$.
Note that in this case $\deg m = \max(\deg m, 2d)$.

In both cases we conclude that $m\equiv_D g \pmod{\cF}$ where $D=\max(\deg m, 2d)$, $g\in k[x_1]$ and $\deg g < \deg f \le d$.
\end{proof}

The following lemma slightly generalizes Lemma 4.4 in \cite{H}.
\begin{lemma}
\label{dp}
Let $k$ be a finite field with $|k| = q=p^m$.  Let $f\in K[x_0]$ with $d=\deg f$ where $K$ is a finite field containing $k$, and $\cF = \{ f, x_0^p-x_1,\ldots, x_{m-1}^p - x_0\}$.
Suppose $g=gcd ( f, x_0^q-x_0 )\in K [ x_0]$.  Then $g\equiv_{dp} 0 \pmod{\cF}$ and $x_i \equiv_{pd} h_i \pmod{\cF}$ with $h_i\in K[x_0]$ and  $\deg h_i < \deg g \le d$ for $i=1,\ldots,m-1$.
\end{lemma}
\begin{proof}  We have $x_0^p \equiv_{\max\{p,d\}} g_1 \pmod{f}$ for some $g_1\in K[x_0]$ with $\deg g_1 < d$.
We have $g_1^p \equiv_{pd} g_2 \pmod{f}$ for some $g_2\in K[x_0]$ with $\deg g_2 < d$.  Inductively we have
$g_i^p\equiv_{pd} g_{i+1} \pmod{f}$ with $g_{i+1}\in K[x_0]$ and $\deg g_{i+1} < d$, for $i=1,\ldots, m-2$.
In particular it follows that $x_0^q - x_0 \equiv h \pmod{f}$ where $h=g_m-x_0$, so $gcd (x_0^q-x_0,f)=gcd(f,h)$.

We have
$$x_1\equiv_p x_0^p \equiv_{\max\{p,d\}} g_1 \pmod{\cF},$$ and
inductively,  from $x_i \equiv_{pd} g_i \pmod{\cF}$ and noting that $x_i^p-g_i^p$ has degree less than $pd$, we have $$x_{i+1} \equiv_p x_i^p \equiv_{pd} (g_i)^p\equiv_{pd} g_{i+1}\pmod{\cF},$$
for $i=1,\ldots, m-2$.
Finally, $x_0 \equiv_p x_{m-1}^p \equiv_{pd} g_{m-1}^p \equiv_{pd} g_m\pmod{\cF}$.
It follows that $h\equiv_{pd} 0 \pmod{\cF}$, consequently $gcd(h,f) \equiv_{pd} 0 \pmod{\cF}$.  Therefore
$g=gcd(x_0^q-x_0, f)\equiv_{pd} 0 \pmod{\cF}$.

Repeat the earlier argument now with $g$ in place of $f$, we have
$x_0^p \equiv_{\max\{p,\deg g\}} h_1 \pmod{g}$ for some $h_1\in k[x_0]$ with $\deg h_1 < \deg g$.
We have $x_1\equiv_p x_0^p\equiv_{\max(p,\deg g)} h_1 \pmod{\cF}$, and since $x_1^p-h_1^p$ has degree less than $pd$, we have
$$x_2 \equiv_p x_1^p \equiv_{pd} h_1^p \equiv_{pd} h_2 \pmod{\cF}$$ for some $h_2\in k[x_0]$ with $\deg h_2 < \deg g$ and $h_1^p \equiv_{pd} h_2 \pmod{g}$.  Inductively we have
$$x_{i+1}\equiv_p x_i^p \equiv_{pd} h_i^p\equiv_{pd} h_{i+1} \pmod{\cF}$$ with $h_{i+1}\in k[x_0]$, $h_i^p\equiv_{pd} h_{i+1}\pmod{g}$ and $\deg h_{i+1} < \deg g$, for $i=1,\ldots, m-2$.
\end{proof}

We will also need the following two results from \cite{H}.
\begin{lemma}
\label{dFpm}
If $\lambda\in Gl_n (k)$ then
$d_{\cF\cup\cE'} = d_{\lambda^*{\cF}\cup\cE'}$ where $\lambda^*(\cF)=\{ F\circ\lambda: F\in\cF\}$.
\end{lemma}

\begin{lemma}
\label{barf}
Let $f_i (x_i)$ be a monic polynomial in $x_i$ of degree $d_i$ over $\bar{k}$, for $i=1,\ldots,n$.  Let $I\subset \bar{k} [x_1,\ldots,x_n]$ be the ideal generated by $f_i (x_i)$, $i=1,\ldots,n$.  Then for all $f\in\bar{k} [ x_1,\ldots, x_n]$, we have
$f\equiv \bar{f} \pmod{I}$ where $\bar{f}\in \bar{k}[x_1,\ldots, x_n]$ with $\deg_{x_i} \bar{f} < d_i$ for all $i$.
Moreover $f\in I$ if and only if $\bar{f}=0$.
\end{lemma}

\section{Last fall degrees of semilocal systems and computations}

\subsection{The case of solving for closed points}
\begin{theorem}
\label{solveG}
Let $k$ be a finite field and let $\bar{k}$ denote the algebraic closure of $k$.
Suppose an ordered set of polynomials $\cG=\{G_1,\ldots,G_m\}$ is $c$-{\em semi-local} with  $\cG = \mu\circ\cF\circ\lambda$ for some $\mu\in Gl_m (k)$, $\lambda\in Gl_n (k)$, and $c$-local $\cF=\{F_1,\ldots, F_m\}\subset k [ x_1,\ldots, x_n]$ such that $\cF$ can be partitioned into $\ell$ subsets $\cF_1$, ..., $\cF_{\ell}$ where the polynomials in $\cF_i$ are polynomials in $x_{ci-i+1},\ldots, x_{ci}$ for $i=1,\ldots,\ell$.
Suppose $|Z(\cG)|=s$.   Suppose for $i=1,\ldots,\ell$, the ideal generated by $\cF_i$ is radical and $d_{\cF_i} \le c'$ for some constant $c'$.  Then
$d_{\cG} \le s+ c'\lceil\log_2 s\rceil $.  Moreover,
the $s$ points in $Z(\cG)$ can be found in time
$n^{O(c')}+ s^{O(c\log s)}$.
\end{theorem}

\begin{proof} Since $d_{\cG}=d_{\cF}$, we will bound $d_{\cG}$ by bounding $d_{\cF}$.

For $i$ such that $|Z(\cF_i)|=1$, let $Z(\cF_i)=(\alpha_{i1},\ldots,\alpha_{ic})$, then $x_{ci-i+1}-\alpha_{i1}$, ..., $x_{ci}-\alpha_{ic}\in V_{\cF_i, c'}\subset V_{\cF, c'}$.  Therefore if $|Z(\cF_i)|=1$ and $m$ is a monomial in $x_{ci-i+1},\ldots, x_{ci}$, then  $m\equiv_{\max(\deg m, c')} \alpha_m \pmod{\cF}$ for some constant  $\alpha_m$.
Let $I_1$ be the set of $i$ such that $| Z(\cF_i )| = 1$.
For any monomial $m$, writing $m=m_1 m_2$ where $m_1$ is a monomial in the set of variables $\{ x_{ci-i+1},\ldots, x_{ci}: i\in I_1\}$, and $m_2$ is a monomial in the remaining variables, then
$m\equiv_{\max(\deg m, c')} \alpha m_2\pmod{\cF}$ where $\alpha$ is a constant.

Let $I$ be the set of $i$ such that $| Z(\cF_i )| > 1$.  For $i\in I$, let $s_i =|Z(\cF_i )|$, and let
$\Delta_i = \max (c', s_i)$.
We have $s=|Z(\cG)|=|Z(\cF)|=\prod_{i\in I} |Z(\cF_i)|=\prod_{i\in I} s_i$.  Since $s_i \ge 2$ for $i\in I$, it follows that $|I|\le \lceil \log_2 s \rceil$, and that
 $\sum_{i\in I} s_i \le \prod_{i\in I} s_i = s$.
Since  $\max (c', s_i)\le c'+s_i$ and $|I|\le\lceil\log_2 s\rceil$, we have
$$\sum_{i\in I} \Delta_i \le \sum_{i\in I} (s_i+c') \le  s+ c'\lceil\log_2 s\rceil.$$

Apply Lemma~\ref{x-g} to such $\cF_i$, we see that there is some linear form $t_i$ in $x_{ci-i+1},\ldots, x_{ci}$ such that $V_{\cF_i, \Delta_i}$ contains  $g_i (t_i)$, and $x_{ci-i+j} - h_{ci-i+j} (t_i)$, $j=1,\ldots,c$, with $\deg g_i = s_i$ and $\deg (h_{ci-i+j}) < \deg g_i$ for $j=1,\ldots,c$.  Write $m_2=\prod_i m_{2,i}$.  By Lemma~\ref{red} we see that for each monomial $m_{2,i}$ in $x_{ci-i+1},\ldots, x_{ci}$, we have $m_{2,i}\equiv_{\max(\deg m_{2,i}, 2\Delta_i)} \gamma_i (t_i)\pmod{\cF}$ for some univariate $\gamma_i$ of degree less than $\deg g_i =s_i$.

For any monomial $m$, we have $m\equiv_{\max(\deg m, D)} h_m\pmod{\cF}$ where $D=\sum_{i\in I} \Delta_i$, and $h_m$ is a polynomial in $t_i$ where $i\in I$ and $\deg_{t_i} h_m < \deg g_i=| Z(\cF_i ) |$ .  It follows that for any polynomial $H\in R$, $H\equiv_{\max(\deg H, D)} h\pmod{\cF}$ for some polynomial $h$ in $t_{i}$, $i\in I$, where $\deg_{t_i} h < \deg g_i$ for all $i\in I$.  By Lemma~\ref{barf}, we see that $H$ belongs to the ideal generated by $\cF$ if and only if $h=0$, if and only if
$H\equiv_{\max(\deg H, D)} 0\pmod{\cF}$.
We conclude that $d_{\cF} \le D \le s+ c'\lceil\log_2 s\rceil $.

To solve for $Z(\cG)$, first observe again for $i$ such that $|Z(\cF_i)|=1$, let $Z(\cF_i)=(\alpha_{i1},\ldots,\alpha_{ic})$, then $x_{ci-i+1}-\alpha_{i1}$, ..., $x_{ci}-\alpha_{ic}\in V_{\cF_i, c'}\subset V_{\cF, c'}$.  There are $n- O(c\log s)$ such variables, since $|I | = O(\log s)$.
Their images under $\lambda^*$ are in $V_{\cG, c'}$ and are linearly independent.  Therefore the linear polynomials in $V_{\cG, c'}$ form a linear subspace of dimension at least $n-O(c\log s)$.    As we have seen, for $i$ such that $| Z(\cF_i)| > 1$, $V_{\cF_i, \Delta_i}$ contains  $g_i (t_i)$, and $x_{ci-i+j} - h_{ci-i+j} (t_i)$, $j=1,\ldots,c$, with $\deg g_i = s_i$ and $\deg (h_{ci-i+j}) < \deg g_i$ for $j=1,\ldots,c$.  Their images under $\lambda^*$ are in $V_{\cG,  \Delta }$, where $\Delta = \max{\Delta_i} \le \max (c', s)$.
To solve for $Z(\cG)$, first we construct a basis $E$ of $V_{\cG,  c'}$, as well as a basis $E_1$ of the subspace consisting of the linear polynomials in $V_{\cG, c'}$.   This can be done in time $n^{O(c')}$ (see \cite{HKY,HKYY}) and note that $E_1$ consists of $n-O(c\log s)$ linear polynomials.  Using the linear equations in $E_1$, $\cG\cup E$ can be reduced to a polynomial system $\cG'$ in $O(c\log s)$ variables.   We are now reduced to solving a polynomial system consisting of a basis of $V_{\cG',  s }$.  The system can be constructed and solved in $s^{O(c\log s))}$ time (see \cite{HKY,HKYY}).
\end{proof}

\subsection{The case of solving for rational points}
Suppose $q=p^d$ where $p$ is prime.  Then $k=\F_q=Z(x^q-x)\simeq Z(\{ x^p-x_1, x_1^p-x_2,\ldots,x_{d-1}^p-x\})$ where $x\in k$ corresponds to $(x,x_1,\ldots,x_{d-1})$, with $x_1=x^p$,\ldots, $x_{d-1} = x_{d-2}^p$.
Throughout this section we identify $k[x_1,\ldots,x_n]$ as a subring of $k[x_{ij}: i=1,\ldots,n, j=0,\ldots,d-1]$ by identifying $x_i$ with $x_{i0}$ for $i=1,\ldots,n$, hence $f(x_1,\ldots, x_n)\in k[x_1,\ldots,x_n]$ with $f(x_{10},\ldots,x_{n0})\in k[x_{ij}: i=1,\ldots,n, j=0,\ldots,d-1]$.
Let $\cE'=\{x_{i0}^p-x_{i1},\ldots, x_{i \ d-1}^p-x_{i0}: i=1,\ldots,n\}$.
Then $Z(\cE') \simeq k^n$ where $(x_{ij})$ corresponds to $(x_{i0})$, and
$Z(\cF \cup \cE')\simeq Z_k (\cF)$.

\begin{theorem}
\label{lfdk}
Let $k$ be a finite field and let $\bar{k}$ denote the algebraic closure of $k$.
Suppose an ordered set of polynomials $\cG=\{G_1,\ldots,G_{n'}\}$ is $c$-{\em semi-local} with  $\cG = \mu\circ\cF\circ\lambda$ for some $\mu\in Gl_{n'} (k)$, $\lambda\in Gl_n (k)$, and $c$-local $\cF=\{F_1,\ldots, F_{n'}\}\subset k [ x_1,\ldots, x_n]$ such that $\cF$ can be partitioned into $\ell$ subsets $\cF_1$, ..., $\cF_{\ell}$ where the polynomials in $\cF_i$ are polynomials in $x_{ci-i+1},\ldots, x_{ci}$ for $i=1,\ldots,\ell$.
Suppose for $i=1,\ldots,\ell$, the ideal generated by $\cF_i$ is radical. Let  $\Delta = \max\{ | Z (\cF_i) |, d_{\cF_i}: i=1,\ldots,\ell\}$. (1) If $|Z_k (\cF) | = 1$ then $d_{\cG\cup\cE'} = O(\Delta p)$.  (2) If $|Z_k (\cF) | =s_0$ and $\left(\begin{array}{c} s\\2 \end{array}\right) < |k|$ where $\ Z(\cF_i ) \le s$ for all $i$,  then $d_{\cG\cup\cE'} = O(\max(\Delta p, s_0 p, s_0 \log s_0))$. (3) The $k$-rational points in $Z(\cG)$ can be found in  $(nd)^{O(\Delta p)} + s_0^{O(\log s_0)}$ time.
\end{theorem}

\begin{proof} Since $\cG = \mu(\lambda^*(\cF))$, $d_{\cG\cup\cE'}=d_{\lambda^* (\cF)\cup\cE'}$.  By Lemma~\ref{dFpm} $d_{\cF\cup\cE'} = d_{\lambda^*(\cF)\cup\cE'}$.  Therefore $d_{\cG\cup\cE'}=d_{\cF\cup\cE'}$, so to prove (1) and (2) of the theorem it is enough to focus on $\cF$.

For $j=ci-i+1,\ldots,ci$, consider $f_j (x_j)=\prod_{\alpha\in Z(\cF_i)} (x_j - x_j (\alpha))$.  Since the ideal generated by $\cF_i$ is radical, $f_j (x_j)$ is in the ideal.   Since $\deg f_j\le |Z(\cF_i)|$,  $f_i (x_i) \in V_{\cF_i, \Delta}$.

Apply Lemma~\ref{dp} to $f_j (x_j)$, letting  $g_j=gcd ( f_j, x_j^q-x_j )$ in $k[x_j]$,  it follows that $g_j\equiv_{s p} 0 \pmod{\cG_j}$
where $\cG_j = \{ f_j, x_j^p - x_{j1}, x_{j1}^p - x_{j2},\ldots, x_{j\ d-1}^p - x_j\}$.   Since $\cG_j\subset V_{\cF_i\cup\cE'_i, \Delta}$, it follows that
$g_j\in V_{\cF\cup\cE', \max(sp,\Delta)}$.  Note that $\max(sp,\Delta) \le \Delta p$, so $g_j\in V_{\cF\cup\cE', \Delta p}$.

If $Z_k (\cF_i)$ has one point then $\deg g_j = 1$ hence $x_j\equiv_1 \beta_j \pmod{g_j}$ where $\beta_j\in k$ for $j=ci-i+1,\ldots, ci$.
If $|Z_k (\cF)|=1$, then $|Z_k (\cF_i)|=1$ for all $i$, so $x_j\equiv_1 \beta_j \pmod{g_j}$ for $j=1,\ldots,n$.  For all $F\in k[x_1,\ldots,x_n]$,
$F(x_1,\ldots,x_n]\equiv_{\deg F} F (\beta_1,\ldots,\beta_n) \pmod{\{g_1,\ldots, g_n\}}$.

We have $x_{j1}\equiv_p x_j^p \equiv \beta_j^p\pmod{\cE'\cup\{g_j\}}$ and it is easy to see inductively $x_{jk}\equiv_p \beta_j^{p^k} \pmod{\cE'\cup\{g_j\}}$.
Put $\beta_{jk}=\beta_j^{p^k}$.  Since $g_j \in V_{\cF\cup\cE', \Delta p}$, we have $x_{jk}\equiv_{\Delta p} \beta_{jk} \pmod{\cF\cup\cE'}$ for all $j,k$.
It follows that for all $F\in k[x_{ij}: i=1,\ldots,n; j=0,\ldots, d-1]$,
$F\equiv_{\max(\Delta p,\deg F)}  F(\beta_{jk}: j=1,\ldots,n, k=0,\ldots,d-1) \pmod{\cF\cup\cE'}$, and $F$ is in the ideal generated by $\cF\cup\cE'$ if and only if $F(\beta_{jk}: j=1,\ldots,n, k=0,\ldots,d-1)=0$.  Therefore $d_{\cF\cup\cE'} \le \Delta p$.  This proves (1).

More generally let $I=\{ i: 1\le i \le \ell, |Z_k (\cF_i ) | \ge 2 \}$.  Let $|I| = e$.  Then $2^e \le s_0$.  Without loss of generality suppose $I = \{ 1,\ldots, e\}$.  Then for all $F\in k[x_{ij}: i=1,\ldots,n; j=0,\ldots, d-1]$,
$F\equiv_{\max(\Delta p,\deg F)}  F(x_{11},...,x_{ce\ d-1},\beta_{ce+1 \ 1},..., \beta_{n \ d-1} ) \pmod{\cF\cup\cE'}$.

Now consider $\cF_i$ for $i=1,\ldots, e$.  Since the ideal generated by $\cF_i$ is radical, by Lemma~\ref{x-g} we know that there is some linear form $t_i=a_{i1} x_{ci-c+1}+\ldots + a_{ic} x_{ci}$ with $a_{i1},\ldots,a_{ic}\in k$, and univariate polynomials $f_i$ and $h_{ci-c+1}$, ..., $h_{ci}$ with coefficients in $K$ where $f_i$ is of degree $| Z(\cF_i) | \le s$, and $h_{ci-c+1}$, ..., $h_{ci}$ are of degree less than $\deg f_i\le s$, such that the ideal generated by $\cF_i$ is also generated by the set $\cF'_i$ consisting of $f_i (t_i)$ and $x_j - h_j (t_i)$ for $j=ci-c+1$, ..., $ci$.
We have $\cF'_i \subset V_{\cF_i, \Delta}$.

Let $t_{ij} = a_{i1}^{p^j} x_{ci-c+1 \ j}+\ldots+ a_{ic}^{p^j} x_{ci \ j}$ for $j=0,\ldots, m-1$.  Then $t_i = t_{i0}$ and $t_{ij}^p \equiv_p t_{i \ j+1} \pmod{\cE'_i} \}$, where $\cE'_i = \{x_{j0}^p-x_{j1},\ldots, x_{j \ d-1}^p-x_{j0}: j=ci-i+1,\ldots, ci\}$.
Let $\cG'_i = \{ f_i(t_i), t_i^p-t_{i1}, \ldots, t_{i\ d-1}^p - t_i\}$.  Then $\cG'_i\subset V_{\cF_i\cup\cE'_i,\max(\Delta,p)}$.

Apply Lemma~\ref{dp} to $\cG'_i$, letting  $f'_i=gcd ( f_i, t_i^q-t_i )\in k [ t_i]$, it follows that $f'_i\equiv_{s p} 0 \pmod{\cG'_i}$ and $t_{ij} \equiv_{sp} h_{ij} \pmod{\cG'_i}$ with $h_{ij}\in K[t_i]$ and  $\deg h_{ij} < \deg f'_i \le s_0$ for $j=1,\ldots,m-1$.

We have $h_j (t_i)\equiv_{\deg h_j} h'_j (t_i) \pmod{f'_i }$ with $\deg h'_j < \deg f'_i \le s_0$ for $j=ci-c+1,\ldots, ci$.  Hence
$x_j - h_j (t_i) \equiv_{\deg h_j} x_j - h'_j (t_i) \pmod{f'_i }$.  Since $x_j \equiv_{\deg h_j} h_j (t_i)\pmod{\cF'_i}$,
$x_j \equiv_{\deg h_j} h'_j (t_i)\pmod{\cF'_i\cup\{ f'_i \} }$.  Since $f'_i\equiv_{sp} 0\pmod{\cG'_i}$ and $\cG'_i \subset V_{\cF_i\cup\cE'_i,\max(\Delta,p)}$, we conclude that $f'_i\equiv_{\Delta p} 0 \pmod{\cF_i \cup \cE'_i}$, and  $x_j\equiv_{\Delta p} h'_j (t_i) \pmod{\cF_i \cup \cE'_i}$ for $j=ci-c+1,\ldots, ci$.

For $j=ci-c+1,\ldots, ci$, we have $x_{j1}\equiv_p x_j^p \pmod{\cE'_i}$, $x_j^p \equiv_{s_0 p} (h'_j (t_i))^p \pmod{x_j - h'_j (t_i)}$, $(h'_j (t_i))^p\equiv_{s_0 p} h'_{j1} (t_i) \pmod{f' (t_i)}$ with $\deg h'_{j1} < \deg f'_i \le s_0$.  We conclude
$x_{j1}\equiv_{\max(\Delta,s_0)p} h'_{j1}\pmod{\cF_i\cup \cE'_i}$
It is easy to see inductively $x_{jk}\equiv_{\max(\Delta,s_0) p} h'_{jk} (t_i) \pmod{\cF_i\cup \cE'_i}$ where
$\deg h'_{jk} < s_0$.

For $i=1,\ldots, e$, let $\Lambda_i$ be the set consisting of $f'_i (t_i)$ and  $x_{jk} - h'_{jk} (t_i)$ for $j=ci-i+1,\ldots, ci$ and $k=0,\ldots,d-1$.   Then $\Lambda_i \subset V_{\cF_i\cup\cE'_i, \max(\Delta, s_0) p}$.
Suppose $m$ is a monomial in $x_{jk}$ where $j=ci-i+1,\ldots, ci$ and $k=0,\ldots,d-1$.  Then
 $m\equiv_{\max(\deg m, 2s_0)} h (t_i) \pmod{\Lambda_i}$ for some $h (t_i)$ of degree less than $s_0$. Since $\Lambda_i \subset V_{\cF_i\cup\cE'_i, \max(\Delta, s_0) p}$, it follows that $m\equiv_{\max(\deg m, \Delta p, s_0 p)} h (t_i) \pmod{\cF_i\cup\cE'_i}$ for some $h (t_i)$ of degree less than $s_0$.

Let $\Lambda=\cup_{i=1}^e \Lambda_i$. Suppose $m$ is a monomial in $x_{ij}$, $i=1,\ldots, ce$, $j=0,\ldots, d-1$.    Then
$m\equiv_{\max(\deg m , 2s_0, s_0 e)} h_m (t_1,...,t_e) \pmod{\Lambda}$ with $\deg_{t_i} h_m < s_0$ and $\deg h_m < s_0 e$.  Hence
$m\equiv_{\max(\deg m , s_0 p, s_0 e, \Delta p)} h_m (t_1,...,t_e) \pmod{\cF\cup\cE'}$ with $\deg_{t_i} h_m < s_0$ and $\deg h_m < s_0 e$.

For $F\in k[x_{ij}: i=1,\ldots,n; j=0,\ldots, d-1]$,  we have
\begin{eqnarray*}
F & \equiv_{\max(\deg F, \Delta p)}&  F(x_{11},...,x_{ce\ d-1},\beta_{ce+1 \ 1},..., \beta_{n \ d-1} ) \pmod{\cF\cup\cE'}\\
 & \equiv_{\max(\deg F, \Delta p, s_0 p, s_0 e)} & F' (t_1,\ldots, t_e) \pmod{\cF\cup\cE'}
\end{eqnarray*}
where $\deg_{t_i} F' < \deg f'_i\le s_0$ for $i=1,\ldots,e$.

If $F$ is in the ideal $J$ generated by $\cF\cup\cE'$ then $F'$ is in the ideal $J\cap k[t_1,\ldots,t_e]$, which is generated by $f'_1 (t_1)$, ..., $f'_e (t_e)$.  By Lemma~\ref{barf} we conclude $F'=0$ since $\deg_{t_i} F' < \deg f'_i$ for $i=1,\ldots, e$.  Therefore if
$F\in J$  then $F\equiv_{\max(\deg F, \Delta p, s_0 p, s_0 e)} 0 \pmod{\cF\cup\cE'}$.  Since $e=O(\log s_0)$, we conclude
that $d_{\cF\cup\cE'} = O(\max(\Delta p, s_0 p, s_0 \log s_0))$, hence (2) follows.

To prove (3), first note that for $i> e $ where $| Z_k (\cF_i ) |=1$, $x_i \equiv_{\Delta p} \beta_i \pmod{ \cF_i \cup \cE_i'}$ for some $\beta_i \in k$, as argued before.  So, the images of $x_i - \beta_i$ under $\lambda^*$ are all in $V_{\cG\cup\cE', \Delta p}$.   They are linearly independent and their number is $n-e \ge n-\log_2 s_0$.

For $i$ such that $| Z_k (\cF_i) | > 1$, we have as argued before $f'_i (t_i)$, $x_j - h'_j (t_i) \in V_{\cG\cup\cE', \Delta p}$, for $j=ci-i+1,\ldots, ci$, $\deg f'_i = |Z_k (\cF_i) |$ and $\deg h'_j < \deg f'_i$ for $j=ci-i+1,\ldots, ci$.

The image under $\lambda^*$ of $x_i - \beta_i$, $i  > e$, and $f'_i (t_i)$, $x_j - h'_j (t_i)$, $j=ci-i+1,\ldots, ci$, $i \le e$, completely determines $Z_k (\cG)$.  They are all in $V_{\cG\cup\cE', \Delta p}$.

Put $U=V_{\cG\cup\cE', \Delta p} \cap k [ x_1,\ldots, x_n]$.  For $i \ge 0$ let $U_i  = \{ H\in U: \deg H \le i\}$.   A basis $B_1$ of $U_1$ as well as a basis $B_2$ of $U_{s_0}$  can be constructed in $(nd)^{\Delta p}$ time.   We have $Z (B_1 ) \subset Z(B_2) = Z_k ( \cG )$.  In solving for $Z( B_2)$ we can use the linear equations in $B_1$ to eliminate $|B_1 | \ge n-\log_2 s_0$ variables.  Then we are left with a system in $O(\log s_0)$ variables of degree bounded by $s_0$.  The reduced system can then be solved in
$s_0^{O(\log s_0)}$ time.  We conclude that the total running time for finding $Z_k (\cG)$ is $(nd)^{O(\Delta p)} + s_0^{O(\log s_0)}$.
\end{proof}

\section{Weil descent systems}
\label{weil}
Throughout this section suppose $k$ is a finite field and $K$ is an extension of degree $n$ over $k$.

Suppose $\theta_1,\ldots,\theta_n$ is a basis of $K$ over $k$.  Let $\bt=(\theta_1,\ldots,\theta_n)$.
Suppose $f\in K[x_1,\ldots,x_c]$.  Let $\hat{x}_i = (x_{i1},\ldots,x_{in})$ for $i=1,\ldots, c$.
For $\hat{x}=(x_1,\ldots,x_n)\in\bar{k}^n$, let $\la \hat{x},\bt\ra = \sum_{i=1}^n x_{i}\theta_j$ for $i=1,\ldots,n$.
The $K$-linear map from $t: k^n\to K$ sending $\hat{x} =(x_1,\ldots,x_n)\in k^n$ to $\la \hat{x},\bt\ra = \sum_{i=1}^n x_i \theta_i$ is bijective.

The {\em Weil decent} of $f$ with respect to $\bt$ consists of an ordered set $\hat{f}$ of $n$ polynomials $f_1,\ldots,f_n\in k[\hat{x}_1,\ldots,\hat{x}_c]$ defined by the following eqution:
\[ f(t (\hat{x}_1),\ldots,t (\hat{x}_c))= f(\la \hat{x}_1,\bt\ra,\ldots,\la \hat{x}_c,\bt\ra)=\sum_{i=1}^n f_i \theta_i = \la \bt, \hat{f}\ra.\]

Let $\sigma$ be the Frobenius automorphism of $K$ over $k$.  Then
\[ (f(t (\hat{x}_1),\ldots,t (\hat{x}_c)))^{\sigma^i}= (\sum_{j=1}^n f_j \theta_j)^{\sigma^i}.\]

We have
\[((\sum_{j=1}^n f_j \theta_j)^{\sigma^i})_{i=0}^{n-1} =(\sum_{j=1}^n f_j \theta_j^{\sigma^i})_{i=0}^{n-1} =
(\la \bt^{\sigma^i}, \hat{f}\ra)_{i=0}^{n-1} = \mu\circ\hat{f},\]
where $\mu\in Gl_n (K)$ with $\bt^{\sigma^i}$ as the $i$-th row, if the rows are indexed by $i=0,\ldots, n-1$.

On the other hand,
\[ ((f(t (\hat{x}_1),\ldots,t (\hat{x}_c)))^{\sigma^i})_{i=0}^{n-1}= (f^{\sigma^i} ( \la \bt^{\sigma^i}, \hat{x}_1\ra,\ldots, \la \bt^{\sigma^i},\hat{x}_c\ra ))_{i=0}^{n-1} =(f^{\sigma^i} (\lambda_i (\hat{x}_1,\ldots,\hat{x}_c))_{i=0}^{n-1}
= (f^{\sigma^i})_{i=0}^{n-1}\circ\lambda,\]
where $\lambda\in Gl_{cn} (K)$ and $\lambda = (\lambda_i)_{i=0}^{n-1}$ where $\lambda_i: \bar{k}^{cn}\to\bar{k}^c$ sending $(\hat{x}_1,\ldots,\bar{x}_c)\in\bar{k}^{cn}$ to $(\la \bt^{\sigma^i},\hat{x}_1\ra,\ldots, \la \bt^{\sigma^i},\hat{x}_c \ra )$.
Note that  up to row permutation $\lambda$ is the block diagonal matrix with $c$ copies of $\mu$ on the diagonal.

Therefore, $\mu\circ\hat{f} = (f^{\sigma^i})_{i=0}^{n-1}\circ\lambda$.

Now suppose $\cF$ is an ordered set of $c$ polynomials $f_1,\ldots,f_c\in K[x_1,\ldots,x_c]$.  Let $\hat{f}_i$ be the ordered set consisting of the Weil descent of $f_i$ with respect to $\bt$ for $i=1,\ldots, c$.  The Weil descent system of $\cF$ with respect to $\bt$, $\hat{\cF}$, consists of the polynomials in $\hat{f}_i$, $i=1,\ldots, c$.  We have
\[ \tilde{\mu}\circ\ \hat{\cF} = (f_1^{\sigma^i},\ldots,f_c^{\sigma^i})_{i=0}^{n-1}\circ\lambda\]
where $\tilde{\mu}\in Gl_{cn} (K)$ and up to row permutation $\tilde{\mu}$ is the block diagonal matrix with $c$ copies of $\mu$ on the diagonal.
Therefore for
\[ \hat{\cF} = \tilde{\mu}^{-1} \circ (f_1^{\sigma^i},\ldots,f_c^{\sigma^i})_{i=0}^{n-1}\circ\lambda\]
and we conclude that $\hat{\cF}$ is a $c$-semi-local polynomial system.

If $\rho\in Gl_{cn} (k)$ then
\[\rho^{-1}\circ \hat{\cF}\circ\rho = \rho^{-1} \circ \tilde{\mu}^{-1} \circ (f_1^{\sigma^i},\ldots,f_c^{\sigma^i})_{i=0}^{n-1}\circ\lambda\circ\rho\]
and we conclude that $\rho^{-1}\circ \hat{\cF}\circ\rho$ is a $c$-semi-local polynomial system.
We have proved the following.

\begin{theorem}
\label{wd}
Suppose $\cF$ is an ordered set of $c$ polynomials $f_1,\ldots,f_c\in K[x_1,\ldots,x_c]$.  Then the Weil descent system $\hat{\cF}$ of $\cF$ is $c$-semi-local.  If $\rho\in Gl_{cn} (k)$ then $\rho^{-1}\circ \hat{\cF}\circ\rho $ is $c$-semi-local.
\end{theorem}

\section{Cryptographic applications}
\label{crypto}
\subsection{Improved attack on HFE}
In the general HFE (Hidden Field Equations) public key cryptosystem, a multivariate public encryption map is constructed by disguising a univariate polynomial over a large finite field in a $K/k$-Weil descent, where $k=\F_q$, $K=\F_{q^n}$, and $n$ is determined by the security parameter.  Let $F\in K[x]$ and let $\hat{\cF}$ be the Weil descent of $\cF=\{F\}$ with respect to a publicly known basis of $K$ over $k$.  Choose random $\lambda, \mu \in Gl_n (k)$.  Let $\cG = \mu\hat{\cF}\lambda$.  In the HFE system $\cG$ can serve as the public encryption function, with $\lambda,\mu$ as the secret decryption key.  The degree of $\hat{\cF}$ needs to be $O(1)$  so that the size of $\cG$ is polynomially bounded.
A well-known choice of $F$ is the {\em extended Dembowski-Ostrom polynomial}
\[ F= \sum_{0\le i\le j < r} a_{ij} x^{q^i+q^j} +\sum_{0\le i < r} b_i x^{q^i} + c \]
with $a_{ij}, b_i, c \in K$.

For cryptanalysis one would like to solve $\cG(x)=\alpha$ given $\alpha\in k^n$.  The solution set is $Z_k (\cG-\alpha) = Z(\{\cG-\alpha\}\cup\cE)$ where $\cE=\{x_i^q-x_i: i=1,\ldots,n\}$ and $k=\F_q$.  The results in \cite{HKY} bounds the last fall degree of $\{\cG-\alpha\}\cup\cE$ by $O(dq)$, leading to an attack that runs in $n^{O(dq)}$ time.  The more recent results in \cite{GMP,HKYY} bounds the last fall degree of  $\{\cG-\alpha\}\cup\cE$ by $O(q\log_q d)$
(with the bound in \cite{GMP} improving on the bound in \cite{HKYY} by a factor of 1/2), leading to an attack that solves the system in $n^{O(q\log_q d)}$ time.
When $d=O(1)$ and $q=O(1)$, the running time is polynomial in $n$ in these attacks.
Our results lead to a further improvement wherever $\cG$ is injective, as elaborated below.

For $\alpha\in k^n$, $\mu\circ(\hat{\cF}-\mu^{-1} (\alpha))\circ \lambda=\cG_{\alpha}$ where $\cG_{\alpha}=\cG-\alpha$. By Theorem~\ref{wd}, $\cG_{\alpha}$ is a 1-semi-local with the local components $F^{\sigma^i} - u^{\sigma^i}_{\alpha}$ where $\sigma$ is the Frobenius map over $k$, and $u_{\alpha}=\rho (\mu^{-1} (\alpha))$.

Note that for $\alpha\in k^n$, $\cG (x)=\alpha$ has a unique solution in $k^n$ if and only if $\hat{\cF} (z)=\mu^{-1} (\alpha)$ has a unique solution in $k^n$, if and only if $F(y)=u_{\alpha}$ has a unique solution in $K$.  In this case,
$F^{\sigma^i} (y)=u^{\sigma^i}_{\alpha}$ has a unique solution in $K$,
so $Z_K (\{ F^{\sigma_i}-u^{\sigma_i}_\alpha\})=1$.
It follows from Theorem~\ref{lfdk} that $d_{\cG_{\alpha}\cup\cE'} = O(dp)$, where $d=\deg F$.  Therefore $\cG(\hat{x}) = \alpha$ can be solved in $n^{O(dp)}$ time.
This improves on the attacks in \cite{GMP,HKY,HKYY}, when $d=O(1)$.

More generally, one can consider using semi-local systems as the basis for constructing public key cryptosystems.  The attacks implied in our results extend to public key encryption functions which are based on semi-local polynomials systems where either the number of closed point solutions is small, or the characteristic of the field is small.  It remains plausible  to construct public key cryptosystems based on semi-local polynomial systems with exponential number of closed point solutions over a finite field of large prime characteristic.   Such a method is discussed below.

\subsection{A public key cryptosystem based on semi-local polynomial systems}
First we describe the basic idea of our construction.
\begin{enumerate}
\item
Choose a finite field $k$.   The characteristic $p$ of $k$ should be large enough, say linear in the security parameter.
\item
Fix some constants $c$ and $r$.  Choose $n=c\ell$ such that $n$ is large enough, say linear in the security parameter. Find a $c$-local polynomial system $\cF$ that can be partitioned into subsystems
$\cF_1,\ldots,\cF_{\ell}$ where the polynomials
$F_{i1},\ldots, F_{ir}$ in $\cF_i$ are in $c$ variables.  We need that each $\cF_i$ defines an injective map from $k^{c}$ to $k^r$, so that $\cF$ defines an injective map from  $k^{n}$ to $k^m$ where $m=r\ell$.  On the other hand we need that $\cF_i$ is not injective as a map from $\bar{k}^{c}$ to $\bar{k}^r$.
\item
Choose random $\lambda\in Gl_n (k)$ and $\mu\in Gl_m (k)$, and form $\cG = \mu\circ\cF\circ\lambda$
\item
Announce $\cG$ as the public encryption map from $k^n$ to $k^m$.
The secret decryption key consists of $\lambda$ and $\mu$
\end{enumerate}

Since $\cF$ is injective from $k^{n}$, the resulting encryption map $\cG:k^{n}\to k^m$ is injective.  Decryption amounts to solving the following problem:  Given $y\in k^m$, to find the unique  $x\in k^{n}$ such that $\cG( x)=y$.

We choose $\cF_i$ of degree $O(1)$ so that the size of $\cG$ is polynomially bounded in $n$.

Suppose $d_{\cF_i} = O(1)$ for all $i$.  If $|Z(\cG)|= O(1)$ then by Theorem~\ref{solveG}, $Z(\cG)$ can be constructed in time polynomial in $n$.  To prevent the attack implicit in Theorem~\ref{solveG}, we want to make sure that in general
there are more than one solutions $x\in\bar{k}^r$ for $\cF_i (x)=\alpha_i$ for all $i$, so that there are in general exponential in $n$ many $x\in\bar{k}^n$ for $\cG (x)=\alpha$.

In light of Theorem~\ref{lfdk} we choose $k$ of large prime characteristic $p$.  More precisely, if $d_{\cF_i} = O(1)$ for all $i$, since $\cF_i$ is injective on $k^c$, $d_{\cG\cup\cE'}=O(p)$.  It will follow that $\cG (x)=\alpha$ can be solved in $n^{O(p)}$ time.  Therefore we want to make sure that $p$ is large.

Let us begin with concrete construction of $1$-semi-local cryptosystems, we can
construct 1-local polynomial system $\cF=\{F_i: i=1,\ldots, n\}$ over $k$, such that $F_i$ is a non-linear univariate permutation polynomial over $k$.  That is, as a univariate polynomial over $k$, $F_i$ is of degree greater than 1, and bijective as a map from $k$ to $k$.    Let $\cF$ be the map from $\bar{k}^n\to\bar{k}^n$, so that $\cF (x_1,\ldots,x_n)  = (F_1 (x_1),\ldots, F_n (x_n))$.

Then
choose random $\lambda, \mu \in Gl_n (k)$ and form $\cG = \mu\circ\cF\circ\lambda$.

Announce $\cG$ as the public encryption map from $k^n$ to $k^n$.
The secret decryption key consists of $\lambda$ and $\mu$.

Perhaps the simplest construction for $F_i$ is $F_i (x_i)=x_i^3$ where $|k^*|$ is not divisible by 3.
the map $x\to x^3$ is bijective on $k$.

Suppose $F_i$ is of degree $d > 1$.  Given $\beta\in\bar{k}^n$ there are $d^n$ solutions to $\cF (x_1,\ldots,x_n) = \beta$ in $\bar{k}^n$, counting multiplicity.

\subsection{Cryptanalysis}
Note that solving the polynomial system $\cG(x)=y$ also reduces to finding a semi-local decomposition $\cG=\mu\circ\cF\circ\lambda$ for $\cG$.   What we have is the semi-local case of the {\em isomorphism of polynomials} (IP) problem: Given two polynomial maps $\cF$ and $\cG$ from $k^n$ to $k^m$, given by two sets of $n$-variate polynomials $\cG = (G_1,\ldots, G_m)$ and $\cF=(F_1,\ldots,F_m)$,  to find
$\lambda\in Gl_n(k)$, $\mu \in Gl_m (k)$ such that $\cG = \mu\circ\cF\circ\lambda$.

As it turns out the square 1-semi-local construction of cryptosystems is vulnerable to a {\em determinant-of-Jacobian} attack developed below.  We will also discuss how partial attacks may be mounted for square $c$-semi-local systems when $c > 1$ by applying the determinant-of-Jacobian method.  We will show how to modify the square system construction into non-square systems to avoid such attacks in the last subsection.

\subsubsection{Determinant-of-Jacobian attack}
Let ${\bf g}= (g_1,\ldots,g_m)$ with $g_i \in k[x_1,\ldots,x_n]$.   We denote by $J({\bf g})$ the $m$ by $n$ Jacobian matrix whose $(i,j)$-th entry is $\partial_j g_i$.

Let $\lambda$ be a linear map from $k^n$ to $k^m$. Let $A_{\lambda}$ denote the $m$ by $n$ matrix representing $\lambda$ such that for $x=(x_i)_{i=1}^n \in k^n$, $\lambda(x)=(\lambda_1 (x),\ldots,\lambda_m (x))$ where $\lambda_i (x)=\sum_{j=1}^n \lambda_{ij}x_j$ and $\lambda_{ij}$ is the $(i,j)$-entry of $A_{\lambda}$.

Below we give a self-contained account of the determinant-of-Jacobian attack \cite{BFP,Kay}.

The following lemma is an easy consequence of linearity in taking partial derivatives.  That is, $\partial_i (af+bg)=a\partial_i f + b \partial_i g$ for $a,b\in\bar{k}$, $f,g\in k[x_1,\ldots,x_n]$.

\begin{lemma}
Let ${\bf f}= (f_1,\ldots,f_n)$ with $f_i \in k[x_1,\ldots,x_m]$, and ${\bf g}= (g_1,\ldots,g_n)$ with $g_i \in k[x_1,\ldots,x_m]$.
Suppose ${\bf g}= \lambda\circ {\bf f}$ where $\lambda$ is a linear map from $k^n$ to $k^n$.  Then
$J({\bf g}) = A_{\lambda} \cdot J({\bf f})$.
\end{lemma}

The following lemma follows from the chain rule: if $f\in k[x_1,\ldots,x_n]$, and ${\bf g}=(g_1,\ldots,g_n)$ where $g_1,\ldots,g_n\in k[x_1,\ldots,x_m]$,  then for $i=1,\ldots, m$, $\partial_i f(g_1,\ldots, g_n)=\sum_{j=1}^n (\partial_j f)\circ {\bf g} \cdot \partial_i g_j$.
\begin{lemma}
Let ${\bf f}= (f_1,\ldots,f_m)$ with $f_i \in k[x_1,\ldots,x_n]$, and ${\bf g}= (g_1,\ldots,g_m)$ with $g_i \in k[x_1,\ldots,x_n]$.
Suppose ${\bf g}= {\bf f}\circ\lambda$ where $\lambda$ is a linear map from $k^n$ to $k^n$.  Then
$J({\bf g}) = J({\bf f})\circ\lambda \cdot  A_{\lambda}$.
\end{lemma}

From the above two lemmas we have the following (see Lemma 18 of \cite{BFP}).
\begin{lemma}
\label{jacobian}
Let ${\bf f}= (f_1,\ldots,f_m)$ with $f_i \in k[x_1,\ldots,x_n]$, and ${\bf g}= (g_1,\ldots,g_m)$ with $g_i \in k[x_1,\ldots,x_n]$.
Suppose ${\bf g}= \mu\circ {\bf f}\circ\lambda$ where $\mu$ is a linear map from $k^m$ to $k^m$ and $\lambda$ is a linear map from $k^n$ to $k^n$.  Then
$J({\bf g}) = A_{\mu} \cdot J({\bf f})\circ\lambda \cdot A_{\lambda}$.
\end{lemma}

Now consider the case of square semi-local polynomial systems, that is, $\cG=\mu\circ\cF\circ\lambda$ in $n$ variables, where $\cF$ is $c$-local, divided into $n/c$ subsystems, each containing $c$ polynomials in $c$ local variables.

Suppose we are given $\cG$ and $c$-local $\cF$ and we want to find $\mu,\lambda\in Gl_n (k)$ such that $\cG=\mu\circ\cF\circ\lambda$.  By Lemma~\ref{jacobian},  $J(\cG) = A_{\mu} \cdot J(\cF)\circ\lambda \cdot A_{\lambda}$, hence
$\det J(\cG) = \det A_{\mu} \cdot \det (J(\cF)\circ\lambda) \cdot \det A_{\lambda}$.  Let $\det A_{\mu}\det A_{\lambda} = \alpha\in k$.  Since $\det (J(\cF)\circ\lambda)=(\det {\cF})\circ\lambda$, we have arrived at a polynomial equivalence problem: given $\det J(\cG)$ and $\det J(\cF)$, to find $\lambda\in Gl_n (k)$ and a constant $\alpha\in k$ such that $\det J(\cG)=\alpha (\det J(\cF))\circ\lambda$.

Note that $\det J(\cG)$ and $\det J(\cF)$ should not be given in dense form in general since their degree is $O(n)$.  Instead, given a polynomial $H$ in $n$ variables of degree $d$ in dense from, an arithmetic circuit of size polynomial in $n$ and $d$ can be constructed for the evaluation of $\det J(\cH)$ on input $x\in k^n$ (see Theorem 17 of \cite{BFP}).

This approach has been applied to solve the isomorphism of polynomials problem efficiently when $\cF$ is special.  For example, the case $\cF = (x_1^d,\ldots, x_n^d)$ is successfully tackled in \cite{BFP} (see also \cite{Kay}).  We will extend the attack to square 1-semi-local polynomial systems (the case $\cF$ is 1-local).

\subsubsection{Determinant-of-Jacobian method for solving square 1-semi-local case of IP problem}
As noted before, solving the polynomial system $\cG(x)=y$ also reduces to solving the isomorphism of polynomials problem with $\cG$ being 1-semi-local: to find $\mu,\lambda\in Gl_n (k)$ such that $\cG=\mu\circ \cF\circ\lambda$.

\begin{lemma}
\label{flambda}
Suppose $h=f\circ\lambda\in k[x_1,\ldots,x_n] $ where $f$ is univariate defined over $k$ and $\lambda$ is a linear form in $x_1,\ldots,x_n$ over $k$.  Then every irreducible factor of $h$ over $k$ is of the form $g\circ\lambda$ where $g$ is an irreducible factor of $f$.
\end{lemma}

\ \\{\bf Proof}  Write $f = \prod_{i=1}^d (x-\alpha_i)$ with $\alpha_i\in\bar{k}$.  Then $f\circ\lambda=\prod_{i=1}^d (\lambda -\alpha_i)$.   So if $g$ is an irreducible factor of $f\circ\lambda$.  Then $g$ is the product of $\lambda-\tau^j (\alpha)$ where $\tau$ is the Frobenius automorphism of $\bar{k}$ over $k$ and $\alpha=\alpha_i$ for some $i$.  Then $g=f'\circ\lambda$ where $f'$ is the product of the conjugate factors $x-\tau^j (\alpha)$, and $f'$ is an irreducible factor of $f$.  The lemma follows.  $\Box$

\begin{lemma}
\label{factor-lambda}
Let $g\in k[x_1,\ldots,x_n]$.  Suppose there exists univariate $f\in k[x]$ and a linear map $\lambda:k^n\to k$ ($\lambda(0) =0$), such that $g=f\circ\lambda$.
Then $\lambda$ is the unique linear factor of $g-g(0)$ with $\lambda(0) = 0$, hence $\lambda$ can be obtained by factoring $g-g(0)$.
\end{lemma}

\ \\{\bf Proof}  We have $f(x)-f(0) =x^e \prod_{i=1}^{d-e} (x-\alpha_i)$ where $d=\deg f$, $e \ge 1$, $\alpha_i\in\bar{k}$ with $\alpha_i\neq 0$ for all $i$.  Since $g(0) = f(\lambda(0))=f(0)$,  we have
$$g-g(0) =f\circ \lambda - f(\lambda (0)) = (f-f(0))\circ\lambda = \lambda^e \prod_{i=1}^{d-e} \lambda-\alpha_i,$$
and the lemma easily follows.  $\Box$

As mentioned before, given $\cG$ in dense from, a polynomial size arithmetic circuit can be constructed for the evaluation of $\det J(\cG)$ on input $x\in k^n$ \cite{BFP}.   Then by applying Kaltofen's algorithm \cite{Kal} with the arithmetic circuit for $\det J(\cG)$ as the input, we can factor $\det (J(\cG))$ over $k$ in random polynomial time.

Let $\cF = (f_1 (x_1),\ldots,f_n (x_n))$.  Then $J(\cF)$ is diagonal with $\partial_i f_i$ as the $(i,i)$-th entry.  Put $h_i=\partial_i f_i$.  Then
$$\det (J(\cF)\circ\lambda)=(\det J(\cF))\circ\lambda = \prod_{i=1}^n h_i (\lambda_i (x_1,\ldots,x_n)).$$
By Lemma~\ref{jacobian}, $\det J(\cG)$ is a constant multiple of $(\det J(\cF))\circ\lambda$.
So by Lemma~\ref{flambda}, every irreducible factor $g$ of $\det J(\cG)$ over $k$ is of the form $h\circ\lambda_i$ for some $i$ with $h$ being an irreducible factor of $h_i$.  By Lemma~\ref{factor-lambda}, up to constant multiple, $\lambda_i$ can be obtained as the unique linear form that appears in the factorization of $g-g(0)$.
It follows that we can obtain $\lambda_1,\ldots, \lambda_n$ up to constant multiples and permutation.

\subsubsection{Square $c$-semi-local case for $c > 1$}
Suppose $n=c m$ and the variables $x_1,\ldots,x_n$ are divided into $m$ groups, the first group $x_1,\ldots, x_c$; the second group $x_{c+1},\ldots, x_{2c}$, and so on.
Suppose $\cF$ is $c$-local, so that $\cF$ as a set is divided $m$ subsets $\cF_1,\ldots,\cF_m$ where $\cF_i$ depends on the $i$-th group of variables.   Let $h_i=\det J(\cF_i)$.  Then $\det J(\cG) = \alpha \prod_{i=1}^m h_i(\lambda_{ci-c+1},\ldots,\lambda_{ci})$ with $\alpha\in k$.  By factoring $\det J(\cG)$ we obtain irreducible factors of these
$h_i(\lambda_{ci-c+1},\ldots,\lambda_{ci})$.  In special cases, we can deduce useful information on the $\lambda_i$'s.
For example, if the first polynomial $f_1$ of $\cF_1$ is univariate then $\partial_i f_1 = 0$ for $i\neq 1$, it follows that $\partial_1 f_1$ is a factor of $h_1=\det J(\cF_1)$.  It then follows that $(\partial_1 f_1)\circ\lambda_1$ is a factor of $\det (\cG)$.  So $\det J(\cG)$ has an irreducible factor of the form $h \circ \lambda_1$ where $h$ is univariate.  After factoring $\det J(\cG)$, we may try every irreducible factor as a candidate and apply the attack for the 1-semi-local case to it to obtain $\lambda_1$ up to a constant multiple.

Thus, for $c > 1$, the determinant-of-Jacobian attack may be mounted to derive partial information about $\lambda$.
It will be interesting to investigate the scope of square $c$-semi-local systems where a full attack can be made.  However we will turn our attention on constructive means for preventing the determinant-of-Jacobian attack.

\subsubsection{Non-square semi-local construction that avoids the determinant-of-Jacobian attack}
We note that the determinant-of-Jacobian attack, which is based on the equality in Lemma~\ref{jacobian}, depends critically on the the system being square.
To avoid such attacks we apply a simple trick to modify a square $c$-semi-local construction into a non-square $c$-semi-local construction, for $c > 1$.

Suppose we have a local $c$-polynomial map that  ${\bf f}: k^c\to k^c$ defined by $c$ homogeneous polynomials $f_1,\ldots,f_c$ of degree $d$ in $c$ variables $x_1,\ldots, x_c$.  Let $h$ be a homogeneous polynomial of degree $d$ in $x_1,\ldots x_c$.    Then the map ${\bf f}': k^c \to k^{c+1}$ sending $x\in k^c$ to $(f_1 (x),\ldots, f_c (x), h(x))$ is injective.  For security consideration we also require $h$ not to be linearly dependent on $f_1,\ldots f_c$.

Suppose $\cF$ is a $c$-local system in $n=mc$ variables divided into $m$ blocks $\cF_i$, $i=1,\ldots, m$, where each $\cF_i$ consists of $c$ homogeneous polynomials of degree $d$ in $c$ local variables $x_{ci-c+1},\ldots,x_{ci}$.  Apply the trick to modify each $\cF_i:k^c\to k^c$ into some $\cF_i': k^c\to k^{c+1}$ by adding a homogeneous $H_i$ of degree $d$ in the same $c$ local variables, we get a non-square $\cF': k^{n}\to k^{n+m}$ defined by $\cF_i'$, $i=1,\ldots,m$.

Below we describe a concrete construction of 2-semi-local $\cG: k^{2m} \to k^{3m}$ where $n=2m$.

As before,
choose a finite field $k$ of characteristic $p$ large enough (say linear in the security parameter), and
choose $n$ that is large enough, say linear in the security parameter.

Assume $|k^*|$ is not divisible by 3, so that $x^3$ defines a bijection on $k$.   Let ${\bf f} = (f_1, f_2)$ where $f_1 (x_1,x_2)=x_1^3$, and $f_2 (x_1,x_2) = x_2^3$.  Then ${\bf f}$ is bijective on $k^2$.  Let $h (x_1,x_2) =x_1^2 x_2 + x_1 x_2^2$.  Then ${\bf f}'=(f_1,f_2,h)$ defines an injective map from $k^2$ to $k^3$.

We construct 2-local polynomial system $\cF$ with $m$ blocks $\cF_1,\ldots,\cF_m$ using ${\bf f}$.  Thus $\cF_i$ consists of $f_1 (x_{2i-1}, x_{2i}), f_2 (x_{2i-1}, x_{2i})$.  Then $\cF$ defines a bijection on  $k^n$.

We construct 2-local polynomial system $\cF'$ with $m$ blocks $\cF'_1,\ldots,\cF'_m$ using ${\bf f}'$.  Thus $\cF'_i$ consists of $f_1 (x_{2i-1}, x_{2i}), f_2 (x_{2i-1}, x_{2i}), h(x_{2i-1}, x_{2i})$.  Then $\cF'$ defines an injective map from $k^n$ to $k^{n+m}$.

Choose random $\lambda\in Gl_n (k), \mu \in Gl_{n+m} (k)$.   Let $\cG = \mu\circ\cF'\circ\lambda$.
Then announce $\cG$ as the public encryption map from $k^{n}$ to $k^{n+m}$.

The secret decryption key consists of $\lambda$ and $\mu$.

From the equality $\cG = \mu\circ\cF'\circ\lambda$ it should be clear why we require $h$ to be linearly independent of $f_1,\ldots,f_c$.  Otherwise since $\mu^{-1} \cG = \cF'\circ\lambda$, such a linear dependence will induce linear conditions on $\mu^{-1}$.

The system defined by $\cG$ is not square, consisting of $n+m$ polynomials in $n$ variables.  So the determinant-of-Jacobian attack cannot be applied directly.

In trying to adapt the attack to the non-square case, we may consider the submatrix $B$ of $A_{\mu}^{-1}$ consisting of the $n=2m$ rows numbered $3i-2$, $3i-1$ for $i=1,\ldots, m$ and note that $B\cG = \cF\circ\lambda$. Note that $J(B\cG)=J(\cF\circ\lambda)=J(\cF)\circ\lambda$ is a diagonal matrix with $3\lambda_i^2$ for $i=1,\ldots,n$ on the diagonal.
However $B$ is unknown, moreover $B$ and $J(\cG)$ are not square matrices.  So not only $B J(\cG)$ cannot be computed, but also $\det (B J (\cG)$, being in this case a complicated function in the entries of $B$ and $J (\cG)$, cannot be split.  Therefore useful information about the determinant cannot be extracted simply from $J(\cG)$.

The simple construction described above resists all the attacks considered in this paper.  It remains an interesting question whether an efficient attack can be found.


\begin{thebibliography}{10}
\bibitem{BET}
{\sc Bettale, L., Faug{\`e}re, J.-C., and Perret, L.}
\newblock Cryptanalysis of {HFE}, multi-{HFE} and variants for odd and even
  characteristic.
\newblock {\em Des. Codes Cryptogr. 69}, 1 (2013), 1--52.


\bibitem{BFP}
{\sc Berthomieu, J., Jean-Charles Faug{\`e}re, J.-C.,  Perret L.}
\newblock Polynomial-time algorithms for quadratic
isomorphism of polynomials: The regular case.
\newblock {\em J. Complexity} 31 (2015), 590-616.



\bibitem{DIN}
{\sc Ding, J., and Hodges, T.~J.}
\newblock Inverting {HFE} systems is quasi-polynomial for all fields.
\newblock In {\em Advances in cryptology---{CRYPTO} 2011}, vol.~6841 of {\em
Lecture Notes in Comput. Sci.} Springer, Heidelberg, 2011, pp.~724--742.

\bibitem{DPT}
{\sc Ding, J., Petzoldt, A., Tao, C.}
\newblock Efficient Key Recovery for All HFE Signature Variants.
\newblock In: {\em Advances in Cryptology – CRYPTO 2021}, vol.~ 12825 of  {\em Lecture Notes in Computer Science}.

\bibitem{DY}
{\sc Ding J., Yang B.Y.}
\newblock Multivariate Public Key Cryptography.
\newblock In: Bernstein D.J., Buchmann J., Dahmen E. (eds) Post-Quantum Cryptography. Springer, Berlin, Heidelberg, 2009, pp.~ 193-241. 


\bibitem{FAU3}
{\sc Faug{\`e}re, J.-C., and Joux, A.}
\newblock Algebraic cryptanalysis of hidden field equation ({HFE}) cryptosystems using {G}r\"obner bases.
\newblock In {\em Advances in cryptology---{CRYPTO} 2003}, vol.~2729 of {\em Lecture Notes in Comput. Sci.} Springer, Berlin, 2003, pp.~44--60.

\bibitem{Ga}
{\sc Gaudry, P.}
\newblock Index calculus for abelian varieties of small dimension and
the elliptic curve discrete logarithm problem.
\newblock {\em Journal of Symbolic Computation}, vol.
44, no.12, (2009), 1690 --1702.

\bibitem{GMP}
{\sc Gorla, E., Mueller, C. Petit, C.}
\newblock Stronger bounds on the cost of computing Groebner bases for HFE systems
\newblock {\em Journal of Symbolic Computation} 109 (2022), 386-398

\bibitem{GRA} {\sc Granboulan, L., Joux, A., and Stern, J.}
\newblock Inverting hfe is quasipolynomial.
\newblock In {\em Advances in Cryptology - CRYPTO 2006, 26th Annual International Cryptology Conference\/} (2006), vol.~4117 of {\em Lecture
Notes in Computer Science}, Springer, pp.~345--356.

\bibitem{H}
{\sc Huang, M.-D. A.}
\newblock On product decomposition
\newblock accepted to appear in {\em Information Processing Letters}, Volume 181, (2023)


\bibitem{HKY}
{\sc Huang, M.-D. A., Kosters, M., and Yeo, S.~L.}
\newblock Last fall degree, HFE, and Weil descent attacks on ECDLP
\newblock{\em Advances in Cryptology -- Proc. 35th Annual Cryptology Conferences (CRYPTO 2015)}, 581-600. (2015)


\bibitem{HKYY}
{\sc Huang, M.-D. A., Kosters, M., Yang, Y., and Yeo, S.~L.}
\newblock On the last fall degree of zero-dimensional Weil descent systems
\newblock {\em J. Symbolic Computation}, Volume 87, 2018, pp.~207-222

\bibitem{Kal}
{\sc Kaltofen, E.}
\newblock Factorization of polynomials given by straight-line programs.
\newblock {\em Randomness and Computation} JAI Press,
1989, pp.~375--412

\bibitem{Kay}
{\sc Kayal, N.}
\newblock Efficient algorithms for some special cases of the polynomial equivalence problem.
\newblock {\em Proceedings of the Twenty-Second Annual ACM-SIAM Symposium on Discrete Algorithms}, Philadelphia, PA, 2011, pp.~ 1409--1421

\bibitem{KR}
{\sc Kreuzer, M., Lorenzo Robbiano, L.}
\newblock {\em Computational Commutative Algebra 1},
Springer, 2000.

\bibitem{LL}
{\sc Lakshman, Y.N., and Lazard D.}
\newblock On the Complexity of Zero-dimensional Algebraic Systems.
\newblock {\em Effective Methods in Algebraic Geometry},
Volume 94 of the series Progress in Mathematics, (1991),  217-225.

\bibitem{P}
{\sc Patarin, J.}
\newblock Hidden Field Equations (HFE) and Isomorphisms of Polynomials
(IP): two families of asymetric algorithms.
\newblock {\em Eurocrypt'96, LNCS 1070}, 1996, pp. 33-46

\bibitem{PQ}
{\sc Petit, C., Quisquater, J.}
\newblock On Polynomial Systems Arising from a Weil Descent.
\newblock {\em Asiacrypt'12, LNCS 7658}, 2012,  pp. 451-466.



\end{thebibliography}
\end{document}